\def\FULLVERSION{}

\documentclass[sigconf]{aamas}  
\usepackage{balance}  

\usepackage{amsmath,amsfonts,amssymb,amstext,amsthm} 
\usepackage{float}
\usepackage{algorithm}
\floatname{algorithm}{Mechanism}

\usepackage{xcolor}
\usepackage{soul}
\usepackage[utf8]{inputenc}

\usepackage{natbib}
\usepackage{graphicx,color}
\usepackage{tabularx}
\usepackage{bbm}

\theoremstyle{definition}

\newtheorem*{conjecture*}{Conjecture}

\newcommand{\vmin}{v^{\min}}
\newcommand{\vmax}{v^{\max}}
\newcommand{\cmin}{c^{\min}}
\newcommand{\cmax}{c^{\max}}
\newcommand{\ghat}{\widehat{G}}
\newcommand{\Gset}{\mathcal{G}}

\newcommand{\step}{\mathbbm{1}}

\newcommand{\uheur}{U_{\text{heur}}}
\newcommand{\ucoop}{U_{\text{coop}}}
\newcommand{\unet}{U_{\text{net}}}
\newcommand{\loss}{L_{\text{net}}}

\newcommand{\ia}{the agent}
\newcommand{\ua}{the principal}
\newcommand{\Ia}{The agent}
\newcommand{\Ua}{The principal}

\def\ProveMechanismB{}

\setcopyright{ifaamas}  
\acmDOI{}  
\acmISBN{}  
\acmConference[AAMAS'19]{Full version of a paper that appeared in the proceedings of the 18th International Conference on Autonomous Agents and Multiagent Systems (AAMAS 2019)}{May 13--17, 2019}{Montreal, Canada}{N.~Agmon, M.~E.~Taylor, E.~Elkind, M.~Veloso (eds.)}  
\acmYear{2019}  
\copyrightyear{2019}  
\acmPrice{}  

\begin{document}
\title{Obtaining Costly Unverifiable Valuations from a Single Agent}
\ifdefined\FULLVERSION
\subtitle{Making an Appraiser Work for You}
\fi

\author{Erel Segal-Halevi}
\affiliation{%
  \institution{Ariel University}
  \city{Ariel} 
  \state{Israel} 
  \postcode{40700}
}
\email{erelsgl@gmail.com}
\author{Shani Alkoby}
\affiliation{%
  \institution{The University of Texas at Austin}
  \city{Austin}
  \state{Texas, USA} 
  \postcode{78712}
}
\email{shani.alkoby@gmail.com}
\author{Tomer Sharbaf}
\affiliation{%
  \institution{Israel Ministry of Finance}
  \city{Jerusalem}
  \state{Israel} 
  \postcode{91950}
}
\email{tomersrb@gmail.com}
\author{David Sarne}
\affiliation{%
  \institution{Bar-Ilan University}
  \city{Ramat-Gan}
  \state{Israel} 
  \postcode{52900}
}
\email{david.sarne@gmail.com}

\begin{abstract}
We consider the problem of a principal who needs to elicit the true worth of an object she owns from an agent who has a unique ability to compute this information. 
The correctness of the information cannot be verified by the principal, so it is important to incentivize the agent to report truthfully.
Previous works coped with this unverifiability by employing two or more information agents and awarding them according to the correlation between their reports.
In this paper we show that even with only one information agent truthful information can be elicited, as long as the object is valuable for the agent too.  In particular the paper introduces a mechanism that, under mild realistic assumptions, is proved to elicit the information truthfully, even when computing the information is costly for the agent. 
Moreover, using this mechanism, the principal obtains the truthful information incurring an arbitrarily small expense beyond whatever unavoidable costs the setting dictates.

\ifdefined\FULLVERSION
\else
\emph{Full version is available at {https://arxiv.org/abs/1804.08314}}
\fi
\end{abstract}

\keywords{Information elicitation; principal-agent; truthful mechanism}  

\maketitle

\section{Introduction}
Often you own a potentially valuable object, such as an antique, a jewel, a used car or a land-plot, but do not know its exact value and cannot calculate it yourself. 
There are various scenarios in which you may need to know the exact object value. For example:
(a) You intend to sell the object and want to know how much to ask when negotiating with potential buyers.
(b) You want to know how much to invest in an insurance policy covering that object. 
(c) The object is a part of an inheritance you manage in behalf of your co-heirs, and you want to prove to them that you manage it appropriately. 
(d) The object is a land-plot that might contain oil, and you want to know whether to invest in developing it.
(e) You are a firm and required by law to include the value of assets you own in your periodic report.
(f) You are a government auctioning a public asset, and want to publish an accurate value-estimate in order to attract more firms to participate in the auction.
Moreover, you may be required by law (or by political pressure from your voters) to obtain and disclose its true value, to avoid accusations of corruption. 

A common solution in these situations is to buy the desired information from an
\emph{agent} with an expertise in evaluating similar objects \cite{alkoby2017s,Moldovanu1998Goethes,alkoby2014manipulating}. Examples of such agents are: antique-shop owners, jewelers, car-dealers, or oil-firms that own nearby land-plots and thus can drill and estimate the prospects of finding oil in your plot.
The problem is that, in many cases, the information is not verifiable: the information buyer (henceforth ``\ua{}'') cannot tell if the information received is correct. 
This results in a strong incentive for \ia{} to provide an arbitrary value whenever the extraction of the true value is costly or requires some effort, knowing \ua{} will not be able to tell the difference.
For example, if an antique agent gives you a low appraisal for an antique object, and you sell it for that low value, you will never know that you were scammed. 

Even if the true value can be verified later on (e.g, due to unsuccessful drilling for oil), 
this might be too late --- the damage due to using the wrong value might be irreversible and \ia{} might be too far away to be punished. Our goal is thus to ---
\begin{quote}
\emph{
--- develop mechanisms that obtain the true value of an object by incentivizing an agent to compute and report it, even when it is \emph{costly} for the agent,
and even when the information is \emph{unverifiable} by the principal.
}
\end{quote} 
The literature on information elicitation usually makes one of two assumptions: either the information is verifiable by \ua{}, or there are two or more information agents such that reports can be compared with peer reports \citep{faltings2017game}. We study a more challenging setting in which the information is unverifiable, and yet there is only a single agent who can provide it. At first glance this seems impossible: how can \ia{} be incentivized to report truthfully if there is no other source of information for comparison?
We overcome this impossibility by allowing, with some small probability, the transfer of the object to the agent for some fee, as an alternative means of compensation (instead of directly paying \ia{} for the information).  This is applicable as long as \ia{} gains value from owning the object, i.e., both capable of evaluating the object and can benefit from owning it.  This is quite common in real life. For example, both the antique-shop owner and the car-dealer, who play the role of \ia{} in the motivating settings above, can provide a true valuation for the car/antique based on their expertise, and can also benefit from owning it (e.g. for resale). Similarly, the oil-firm who owns nearby sea-plots has access to  relevant information enabling it to calculate the true value of the plot in question, and will also benefit from owning that plot.

We use this principle of selling the object to \ia{} with a small probability for designing a mechanism for eliciting the wanted information. Our analysis proves that the mechanism is truthful whenever the information computation cost is sufficiently small relative to the object's expected value; the exact threshold depends on the prior distribution of the value. For example, when the object value is distributed uniformly, the mechanism is truthful whenever the computation cost is less  than $1/4$ of the expected object value, which is quite a realistic assumption
(Section \ref{sec:prod-costs}).

While our mechanism allows \ua{} to learn the true information, this information does not come for free: \ua{} ``pays'' for it by the risk of selling the object to \ia{} for a price lower than its value.
Our next goal, then, is to \emph{minimize \ua{}'s loss} subject to the requirement of true information elicitation. 
We show that our mechanism parameters can be tuned such that \ua{}'s expected loss is only slightly more than its computation cost.
This is an optimal guarantee, since \ua{} could not get the information for less than its computation cost even if she had the required expertise herself
(Section \ref{sub:min-loss}). 

We then show how the mechanism can be augmented to handle some extensions of the basic model.
These include the case where the object is divisible, when the delivery of information is costly, 
when the principal and the agent have different valuations for the object,
and when the cost the agent incurs when computing the information is unknown to the principal
(Section \ref{sub:variants}).

In addition to our main mechanism, we present two alternative mechanisms that differ in their privacy considerations, in the sequence of roles (in the resulting Stackelberg game), and in the nature of the decisions made by the different players.
Interestingly, despite their differences, all three mechanisms are equivalent in terms of the guarantees they provide.
This leads us to conjecture that these guarantees are the best possible
(Section \ref{sec:alternatives}). 
 
Previous related work is surveyed in Section \ref{sec:Related}. 
Discussion and suggested extensions for future work are given in Section \ref{sec:future}.

\section{The Model}
\label{sec:model}
There is a \emph{principal}
who needs to know the true value of an object she owns or a business opportunity available to her.
The monetary value of this object or opportunity for \ua{} is denoted by $v$. While \ua{} does not know $v$, she has a prior probability distribution on $v$, denoted by $f(v)$, defined over the interval $[\vmin,\vmax]$, with $0\leq\vmin<\vmax$.

\Ua{} can interact with a single \emph{agent}.
Initially, \ia{} too does not know $v$ and knows only the prior distribution $f(v)$. However, \ia{} 
has a unique ability to compute $v$, by incurring some cost $c\geq 0$. Initially we assume that the cost $c$ is common knowledge; in Section \ref{sec:unknown-cost} we relax this assumption.

\Ua{}'s goal is to incentivize \ia{} to compute and reveal the true $v$. However, \ua{} cannot verify $v$ and has no other sources of information besides \ia{}, so the incentives cannot depend directly on whether $v$ is correct. 

Similar to the \emph{common value} setting studied extensively in auction theory \citep{kagel1986winner}, our model assumes that the value of the object to both the principal and the agent is the same. 
A realistic example of this setting is when both \ua{} and \ia{} are oil firms: \ua{} owns an oil field but does not know its value, while \ia{} owns nearby fields and can gain information about the oil field from its nearby drills.
In Section \ref{sub:different-values} we relax this assumption and allow the two values to be different. 

The mechanism-design space available to \ua{} includes transferring the object/opportunity to \ia{}, as well as offering and/or requesting a payment to/from \ia{}. 
\Ua{} has the power to commit to the mechanism rules, i.e, \ua{} is assumed to be truthful. The challenge is to design a mechanism that will incentivize truthfulness on the side of \ia{} too.

\Ia{} is assumed to be risk-neutral and have quasi-linear utilities.
I.e, the utility of \ia{} from getting the object for a certain price is the object's value minus the price paid. If \ia{} calculates $v$, then the cost $c$ is subtracted from his utility too.

The primary goal of \ua{} is to elicit the exact value $v$ from \ia{}. Subject to this, she wants to minimize her expected \emph{loss}, defined as the object value $v$ (if transferred to \ia{}) minus the payments received.

\section{\large Truthful Value-Elicitation Mechanism}
\label{sec:prod-costs}
The mechanism most commonly used in practice for eliciting information is to pay \ia{} the cost $c$ (plus some profit margin) in money.  
However, when information is not verifiable, monetary payment alone cannot incentivize the agent to actually incur the cost of calculation and report the true value. 

Instead, in our mechanism, \ua{} ``pays'' to \ia{} by transferring the object to \ia{} with some small probability. 
The mechanism guarantees \ia{}'s truthfulness, meaning that, under the right conditions (detailed below), a rational agent will choose to incur all costs related to computing the correct value, and report it truthfully. The mechanism is presented as Mechanism  \ref{alg:2ndprice}. It is parametrized by a small positive constant $\epsilon$, and a probability distribution represented by its cumulative distribution function $G$. 

\begin{algorithm}
\begin{enumerate}
	\item \Ua{} secretly selects $r$ at random, distributed in the following way:
		\begin{itemize}
		\item With probability $\epsilon$, this $r$ is distributed uniformly in $[\vmin,\vmax]$;
		\item With probability $1-\epsilon$, this $r$ is distributed like $G(r)$.
		\end{itemize}
	\item \Ia{} bids a value $b$. 
	\item \Ua{} reveals $r$ and then:
	\begin{itemize}
		\item If $b\geq r$, 
		\ua{} gives the object to \ia{}, and \ia{} pays $r$ to \ua{}.
		\item If $b<r$,  no transfers nor payments are made.
	\end{itemize}
\end{enumerate}
\caption{Parameters: $\epsilon>0$: a constant, $G(\cdot)$: a cdf.
\label{alg:2ndprice}
}
\end{algorithm}
The underlying idea is to make \ia{} ``feel like'' in a Vickrey auction. For \ia{}, the random price $r$ is just like a second price in a Vickrey auction; therefore, if \ia{} knows $v$, it is optimal for him to bid $b=v$.
The challenge is to show that it is optimal for \ia{} to actually calculate $v$. This crucially depends on the selection of the cdf $G(r)$. Below we prove a necessary and sufficient condition for the existence of an appropriate $G(r)$.
We assume throughout the analysis that $\epsilon$ is positive but infinitesimally small (i.e, $\epsilon\to 0$).

\begin{theorem}
\label{trm:1}
There exists a function $G(r)$ with which Mechanism \ref{alg:2ndprice} is truthful, if-and-only-if $c<E_v[\max(0, v - E[v])]$.
One cdf with which the mechanism is truthful in this case is:
\begin{align}
\label{eq:g*}
G^*(r) = \step_{r>E[v]}
\end{align}
\end{theorem}
Before proving the theorem, we illustrate its practical applicability with some examples.
\begin{example}
\label{exm:applicability}
$v$ is uniform in $[0,2 M]$. Then $E[\max(0,v-E[v])]=M/4$ so Mechanism \ref{alg:2ndprice} is applicable iff $c < M / 4$ --- the object's appraisal cost should be less than one quarter of the object's expected value.
With the function $G^*$ in \eqref{eq:g*},
Mechanism \ref{alg:2ndprice} selects $r$ in the following way: with probability $\epsilon$, $r$ is selected uniformly at random from $[\vmin,\vmax]$; with probability $1-\epsilon$, $r = E[v]$.
\end{example}
\begin{example}
$v$ has a symmetric triangular distribution in $[0,2 M]$ with mean $M$.
Here, Mechanism \ref{alg:2ndprice} is applicable iff $c < M / 6$.
\end{example}
Both these conditions are realistic, since usually the cost of appraising an object is at least one order of magnitude less than the object's expected value.
For example, 
used cars usually cost tens of thousands of dollars (even very cheap ones cost at least \$2000), and 
the cost of a pre-purchase car inspection generally ranges from \$100 to \$200. Similarly, a used engagement ring typically costs thousands of dollars, while 
hourly rates of a diamond ring appraisals range from \$50 to \$150.

\ifdefined\FULLVERSION
\begin{itemize}
\item If $v$ is distributed uniformly in $[0,M]$, then Mechanism \ref{alg:2ndprice} is applicable iff $c < M / 8$.
\item If $v$ is a discrete random variable that equals $0$ with probability $q$ and $M$ with probability $1-q$, then Mechanism \ref{alg:2ndprice}
 is applicable iff  $c < q(1-q)M$. 
\item If $v$ is distributed exponentially with mean $\lambda$, then the mechanism is applicable iff 
$c < \lambda / e$  (where $e\approx 2.7..$).
\item If $v$ is distributed exponentially with mean $\lambda$, then the mechanism is applicable iff 
$c < \lambda / e$  (where $e\approx 2.7..$), which can be quite high.
\end{itemize}
It is reasonable to assume that the cost $c$ of calculating an object's value is at least one order of magnitude lower than its maximum possible value (for example, the cost of appraising a car is at most several hundreds while the maximum possible value of a car might be hundreds of thousands).
Thus, the mechanism is applicable in the uniform and exponential case, and also in the discrete case when $q\in[0.12,0.88]$.
\fi

\begin{proof}[Proof of Theorem \ref{trm:1}]
\Ia{} has essentially two possible strategies. We call them, following \citet{faltings2017game}, ``cooperative'' and ``heuristic'':
\begin{itemize}
	\item In the cooperative strategy, \ia{} computes $v$ and uses the result to determine a bid $b(v)$.
	\item In the heuristic strategy, \ia{} does not compute $v$, and determines $b$ based only on the prior $f(v)$.
\end{itemize}
\Ia{} will use the cooperative strategy iff its expected utility is larger than the expected utility of the heuristic strategy by more than $c$. Therefore in the following paragraphs we calculate the expected utility of \ia{} in each strategy, showing that under the condition given in the theorem there always exists a cdf $G$ for which the above holds and otherwise the condition cannot hold. For the formal analysis we denote by $\ghat$ the integral of $G$: $\ghat(v) := \int_{r=0}^v G(r)dr$.

In the \emph{cooperative} strategy, \ia{} gets the object iff $r \leq b(v)$, and then his utility is $v-r$. Therefore his expected utility is: 
\footnote{
We consider only cdfs $G$ that are continuous and differentiable almost everywhere, so $G'$ is well-defined almost everywhere. In points in which $G$ is discontinuous (i.e., has a jump), $G'$ can be defined using Dirac's delta function.
}
\begin{align*}
\int_{r=0}^{b(v)} \bigg[(1-\epsilon)G'(r) + {\epsilon\over \vmax-\vmin}\bigg] (v-r)dr
\end{align*}
The integrand is positive iff $r<v$. Therefore the expression is maximized when $b(v)=v$, and hence it is a strictly dominant strategy for \ia{} to bid $v$. In this case, his utility is 
$(1-\epsilon)\int_{r=0}^{v}G'(r)(v-r)dr
 + 
{\epsilon\over \vmax-\vmin}\int_{r=0}^{v}(v-r)dr
$. We assume that $\epsilon\to 0$,%
\footnote{
If $\epsilon=0$, then reporting the true $v$ is only a weakly-dominant strategy: \ia{} never gains from reporting a false value, but may be indifferent between false and true value. For example, if the true value is $2$ and the cdf is uniform in $[3,5]$ and zero elsewhere, then \ia{} is indifferent between reporting $1$ and reporting $2$, since in both cases he loses the object with probability 1.  Making $\epsilon$ even slightly above 0 prevents this indifference and makes reporting $v$  strictly better than any other strategy.

However, to attain this strict-truthfulness, it is sufficient to have $\epsilon$ arbitrarily small. Hence, in the following analysis we assume for simplicity that $\epsilon\to 0$.
}
so that the gain is approximately $\int_{r=0}^{v}G'(r)(v-r)dr$.
By integrating by parts, one can see that this expression equals $\ghat(v)$. 
Hence, before knowing $v$, the expected utility of \ia{} from the cooperative strategy, denoted $\ucoop(G)$, is:
\begin{align*}
\ucoop(G) = E[\ghat(v)]
\end{align*}
where $E$ denotes expectation taken over the prior $f(v)$.

In the \emph{heuristic} strategy, \ia{}'s expected utility as a function of the bid $b$ is:
\begin{align*}
&E\left[\int_{r=0}^{b} [(1-\epsilon)G'(r) + {\epsilon\over \vmax-\vmin}](v-r)dr\right]
\\
=&
\int_{r=0}^{b} [(1-\epsilon)G'(r) + {\epsilon\over \vmax-\vmin}](E[v] -r)dr.
\end{align*}
The integrand is positive iff $r< E[v]$, so it is a strictly dominant strategy for \ia{} to bid $b=E[v]$. In this case, his gain when $\epsilon\to 0$, denoted $\uheur(G)$, is:
\begin{align*}
\uheur(G) = \ghat(E[v])
\end{align*}
The \emph{net utility} of \ia{} from being cooperative rather than heuristic is the difference:
{\small
\begin{align}
\label{eq:unet}
\unet(G) := \ucoop(G) - \uheur(G)
= E[\ghat(v)] - \ghat(E[v])
\end{align}
}
The mechanism is truthful iff $\unet(G) > c$, i.e, the net utility of \ia{} from being cooperative is larger than the cost of being cooperative.  Therefore, it remains to show that there is a cdf $G$ satisfying $c < \unet(G)$, iff $c < E_v[\max(0, v-E[v])]$.
This is equivalent to showing that $E_v[\max(0, v-E[v])]$ is the maximum possible value of $\unet(G)$, over all cdfs $G$. This is a non-trivial maximization problem since we have to maximize over a set of functions. 
We first present an intuitive solution and then a formal solution.

Intuitively, to maximize $E[\ghat(v)] - \ghat(E[v])$ we have to make $\ghat(E[v])$ as small as possible, and subject to that, make $\ghat$ as large as possible. The smallest possible value of $\ghat$ is 0, so we let $\ghat(E[v])=0$. 
Therefore we must have $G(r)=0$ for all $r\leq E[v]$.
Now, to make $\ghat$ as large as possible, we must let it increase at the largest possible speed from $E[v]$ onwards; therefore we must make its derivative $G$ as large as possible, so we let $G(r)=1$ for all $r>E[v]$. All in all, the optimal $G$ is the step function: $G^*(r) = \step_{r>E[v]}$, which gives $\unet(G^*) = E[\max(0,v-E[v])]$ as claimed.

To prove this formally, we use mathematical tools that have been previously used in the analysis of revenue-maximizing mechanisms \citep{manelli2007multidimensional,hart2017approximate}. In particular, we use Bauer's maximization principle:
\begin{quote}
In a convex and compact set, every linear function has a maximum value, and it is attained in an extreme point of the set --- a point that is not the midpoint of any interval contained in the set.
\end{quote}
Denote by $\Gset$ the set of all cumulative distribution functions with support contained in $[\vmin,\vmax]$.
$\Gset$ is a \emph{convex} set, since any convex combination of cdfs is also a cdf. Moreover, it is \emph{compact} with respect to the sup-norm topology (see \citet{manelli2007multidimensional}).
The objective function $\unet(\cdot)$ is  \emph{linear}.
Therefore, to find its maximum value it is sufficient to consider the extreme points of $\Gset$.
We claim that the only extreme points of $\Gset$ are \emph{0-1 step functions} --- functions $G$ for which  $G(r)\in\{0,1\}$ for all $r$. Indeed, suppose that $G$ is not a step function, so there is some $r_0$ for which $0 < G(r_0) < 1$.
Then the following two functions are different elements of $\Gset$:
\begin{align*}
G_1(r) &= \min(1, 2 G(r))
\\
G_2(r) &= \max(0, 2 G(r) - 1)
\end{align*}
$G$ is the midpoint of the segment $G_1$--$G_2$ (see the figure below), so $G$ is not an extreme point of $\Gset$.

\begin{center}
\includegraphics[height=.5\columnwidth,width=.7\columnwidth]{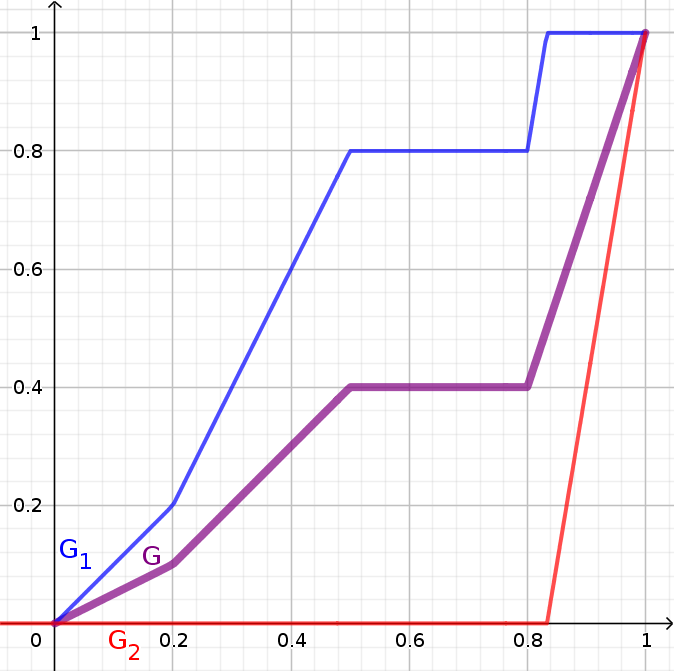}
\end{center}

Therefore, it is sufficient to maximize $\unet$ on cdfs of the following form, for some parameter $t\in[\vmin,\vmax]$:
\begin{align*}
G_t(r) =  \step_{r>t}
\end{align*}
Integrating $G_t(r)$ yields $\ghat_t(v) = \max(0, v-t)$, so by \eqref{eq:unet}:
\begin{align*}
\unet(G_t) = 
&
E[\max(0, v - t)]
-
\max(0, E[v] - t)
\\
=
&
\int_{v=t}^{\vmax} (v - t) f(v) dv
+
\min(0, t - E[v])
\end{align*}
To find the $t$ that maximizes $\unet(G_t)$, we take its derivative with respect to $t$:
\begin{itemize}
\item The derivative of the leftmost term is $-\int_{v=t}^{\vmax}f(v) dv = -\Pr[v\geq t]$, which is always between $0$ and $-1$.

\item The derivative of the rightmost term is $1$ when $t < E[v]$ and $0$ when $t > E[v]$.
\end{itemize}
Therefore, $\unet(G_t)$ is increasing when $t < E[v]$ and decreasing when $t > E[v]$. Therefore its maximum is attained for $t = E[v]$ and it is $E[\max(0,v-E[v])]$ as claimed.
\end{proof}

\section{Minimizing \ua{}'s loss}
\label{sub:min-loss}

The function $G^*$ from the proof of 
Theorem \ref{trm:1}
allows \ua{} to elicit true information
for a large range of costs.
However, the information does not come for free: \ua{} ``pays'' for the information by the possibility of giving away the object. 
As stated earlier, obtaining the information is mandatory.
Hence, 
\ua{} naturally seeks to minimize the loss resulting from giving away the object.
In other words, from the set of all cdfs with which Mechanism \ref{alg:2ndprice} is truthful, \ua{} would like to choose a single cdf $G$ (possibly different than $G^*$) which minimizes her loss.

\Ua{} loses utility whenever \ia{} gets the object, i.e., whenever \ia{} bids $v>r$. In this case, \ua{}'s net loss is $v-r$. Therefore, the expected loss of \ua{} (when $\epsilon\to 0$), denoted $\loss(G)$, is:
\begin{align*}
\loss(G) := E_v\left[
\int_{r=0}^{v} G'(r)(v-r)dr
\right]
\end{align*}
A simple calculation shows that this expression equals $E_v[\ghat(v)]$, which is exactly $\ucoop(G)$ --- the utility of \ia{} from playing the cooperative strategy. This is not surprising as \ia{} and \ua{} are playing a zero-sum game.

To induce cooperative behavior, $\ucoop(G)$ must equal $c'$, for some $c' > c$. Therefore \ua{}'s loss must be $c'$ too. 
Fortunately, \ua{} can attain this loss even for $c'$ arbitrarily close to $c$.  We define:
\begin{align}
\notag
p_{c'} &:= {c'\over E_v[\max(0, v - E[v])]}
\\
\label{eq:gc}
G_{c'}(r) &:= p_{c'}\cdot \step_{r>E[v]} 
+
(1-p_{c'})\cdot \step_{r=\vmax}
\end{align}
With this $G_{c'}$, Mechanism \ref{alg:2ndprice} selects $r$ as follows. With probability $p_{c'}$, $r$ is selected using the function $G^*$ from \eqref{eq:g*}, which guarantees \ia{} a utility of $E_v[\max(0, v - E[v])]$. 
With probability $1-p_{c'}$, \ua{} chooses $r$ so large that \ia{} never gets the object. Therefore the expected utility of \ia{} is $p_{c'}\cdot E_v[\max(0, v - E[v])] = {c'}$.
Consequently the loss of \ua{} is ${c'}$ too. When ${c'}\to c$, \ua{}'s loss approaches the theoretic lower bound --- she obtains the information for only slightly more than its computation cost.

The probability that \ua{} has to sell the object is quite low in realistic settings. For example, consider the case mentioned in Example \ref{exm:applicability} where $v$ is uniform in $[0,2 M]$. Suppose the object is a car. 
Typical values are $c=$\$200
and $M=$\$40K. Here $p_c = c/(M/4) = 0.02$. So the information is computed and delivered truthfully with probability 1, whereas with probability 98\% \ua{} keeps the car (and incurs no loss), and with probability 2\% \ua{} sells the car for its expected value of $\$40K$ (and loses the difference $v-\$40K$).

An interesting special case is $c=0$, i.e., \ia{} already knows $v$.
In this case, \ua{} does not need to know the prior distribution $f(v)$ and can simply use
\begin{align}
\label{eq:g0}
G_0(r) := \delta + (1-\delta)\cdot \step_{r=\vmax}
\end{align}
If $\delta>0$, \ia{}'s net utility is positive so the mechanism is truthful; when $\delta\to 0$, \ua{}'s loss approaches 0.

\section{Variants and Extensions}
\label{sub:variants}
\ifdefined\FULLVERSION
The proposed mechanism can be augmented in various ways to support extensions of the underlying model. 
\subsection{Divisible objects}
So far we assumed the object is indivisible, so it is either sold entirely to \ia{} or not at all.
In general the object may be divisible. For example, it may be possible to sell only a part of an oil field, or only sell shares in the field's future profits. 
In this case, instead of using the function $G_{c'}(r)$ of \eqref{eq:gc}, 
we can run Mechanism \ref{alg:2ndprice} using the function $G^*(r)$, but if the object needs to be sold in step 3 --- only a fraction $p_{c'}$ of it is actually sold.
The analysis of this mechanism is the same --- cooperation is still a dominant strategy for \ia{} whenever $c<c'$, and \ua{}'s loss is still $c'$.
The advantage is that a risk-averse agent may prefer to buy a fraction $p_{c'}$ of the object with certainty, than to buy the entire object with probability $p_{c'}$.
Moreover, a risk-averse principal may prefer to keep a fraction $1-p_{c'}$ of the object with certainty, than to risk selling the entire object with probability $p_{c'}$.

\subsection{Cost of information delivery}
So far we assumed that the information production is costly, but once the information is available --- its delivery is free. 
In general, the information delivery might also be costly. For example, \ia{} might have to write a detailed report about how $v$ was calculated, and have the report signed by the firm's accountants.
\Ia{} incurs the cost of delivery whether the delivered information is true or false. 
\Ua{} can handle this case 
by promising \ia{} 
to pay the delivery cost if \ia{}
participates in Mechanism
 \ref{alg:2ndprice}. 
This makes the strategic situation of \ia{} identical to the situation analyzed in 
Section \ref{sec:prod-costs}. Mechanism \ref{alg:2ndprice} is still truthful whenever $c<E_v[\max(0,v-E[v])]$.
\Ua{}'s loss is the sum of the production and delivery costs, which is the smallest loss possible since \ia{}'s expenses must be covered.
\else
It is straightforward to extend our mechanism to handle divisible objects and cost of information delivery. 
Below we present some additional possible extensions.
\fi

\subsection{Different values}
\label{sub:different-values}
So far we assumed that the object's value is the same for \ua{} and \ia{}. In general the values might be different. 
For example, suppose \ua{} the owner of a used car and \ia{} is a car dealer.  While the dealer certainly has a positive value for owning the car, it may be lower (in case the dealer already has several cars of the same model) or higher (in case the dealer can fix the car at a reduced cost) than its value for the car owner.

Let $v_p$ be the object's value for \ua{} and $v_a$ its value for \ia{}. 
\Ia{}'s utility is calculated as before, using $v_a$ instead of $v$:
\begin{align}
\notag
\ucoop(G) = E[\ghat(v_a)]
\hskip 1cm
\uheur(G) = \ghat(E[v_a])
\\
\label{eq:ucoop-uheur-unet}
\unet(G) = E[\ghat(v_a)] - \ghat(E[v_a])
~~~~~~~~~~
\end{align}
Therefore, 
using the cdf $G^*(r) := \step_{r>E[v_a]}$,
Mechanism \ref{alg:2ndprice} 
is still truthful and elicits $v_a$, as long as the condition of Theorem \ref{trm:1} holds on $v_a$, i.e.:
\begin{align*}
c < E[\max(0, v_a-E[v_a])
\end{align*}
If $v_a$ and $v_p$ are correlated,
then \ua{} can use the knowledge of $v_a$ to gain some knowledge about $v_p$.

However, in contrast to the common-value setting, the game here is no longer zero-sum --- \ua{}'s loss does not equal \ia{}'s utility, so it may be larger or smaller than $c$.
\Ua{}'s expected loss is now:
\begin{align}
\notag
\loss(G) &:= E\left[
\int_{r=0}^{v_a} G'(r)(v_p-r)dr
\right]
\\
\notag
&=
E\big[\ghat(v_a) + (v_p-v_a)\cdot G(v_a)\big]
\\
\label{eq:loss}
&=
\ucoop(G) + E\big[ (v_p-v_a)\cdot G(v_a)\big]
\end{align}
So, when $v_p > v_a$  \ua{}'s loss is larger than \ia{}'s utility, and when $v_p < v_a$ it is smaller.

Suppose we want \ia{}'s net utility to be at least $c'$, for some $c'>c$.
Then, \ua{}'s minimization problem is:
\begin{align}
\notag
\text{minimize}&& E\big[\ghat(v_a)] + E\big[ (v_p-v_a)\cdot G(v_a)\big]
\\
\text{subject to}&&G\in\Gset \text{~~and~~} E[\ghat(v_a)] - \ghat(E[v_a]) \geq c'
\notag
\end{align} 

This is still a problem of minimizing a linear goal over a convex set of functions, so the minimum is still attained in the extreme points of the set. However, finding the extreme points and minimizing $G$ over that points is much harder. We leave it to the future.\footnote{
To get an idea of the loss magnitude,
consider a special case in which $v_a$ and $v_p$ are fully correlated: suppose there is a constant $a$ such that $v_a = a\cdot v_p$. Suppose also, for the sake of the example, that $v_p$ is distributed uniformly in $[0,2 M]$. 
Suppose \ua{} uses 
Mechanism \ref{alg:2ndprice} with the $G_{c'}$ of \eqref{eq:gc}.
Then, using the expressions in the text body, we find that \ua{}'s loss is at most $(3/a - 2) \cdot c'$.
So when $a=1$ \ua{}'s loss is exactly $c'$ (which may be very near $c$), but when $a<1$ the loss is more than $c'$ and when $a>1$ the loss is less than $c'$, as can be expected. It is interesting that the loss (when $a$ is fixed) is  linear function of $c'$. We do not know if this is true in general.
}

\ifdefined\fullversion
\fi

\subsection{Unknown Cost of Computation}
\label{sec:unknown-cost}
So far, we assumed that 
\ua{} knows the costs incurred by \ia{}.
This assumption is realistic in many cases.  For example, when the object is a car, the mechanic can reveal its condition by running a set of standard checks that consume a known amount of time, so their cost can be reasonably estimated. 
However, in some cases 
the cost might be known only to \ia{}.
In this subsection we assume that \ua{} 
only knows a prior distribution on $c$, 
given by pdf $h$ and cdf $H$, with support $[\cmin,\cmax]$. For simplicity we consider here the common value setting, $v_p = v_a = v$.

If \ua{} must get the information at all costs (e.g., due to regulatory requirements), then she can simply run Mechanism \ref{alg:2ndprice} with 
the cdf $G_{c'}$ of \eqref{eq:gc}, taking $c' = \cmax$. This ensures that \ia{} calculates and reports the true information whenever $\cmax < E[\max(0, v-E[v])]$, and \ua{}'s loss is $\approx \cmax$.

However, in some cases \ua{} might think that $\cmax$ is too much to pay for the information. In this case, it may be useful for \ua{} to determine a \emph{utility} of obtaining the information. We denote \ua{}'s utility from knowing the information by $u$, and assume that it is measured in the same monetary units as the function $\loss$ of Section \ref{sub:min-loss}. In other words, \ua{}'s loss is:
\begin{itemize}
\item $\loss(G)$ --- when she elicits the true value using Mechanism \ref{alg:2ndprice} with cdf $G$;
\item $u$ --- when she does not elicit the true value.
\end{itemize}
If $u<\cmax$, it is definitely not optimal for \ua{} to use Mechanism \ref{alg:2ndprice} 
with the cdf $G_{c'}$ taking $c'=\cmax$.
What \emph{should} \ua{} do in this case?

To gain insight on this situation, we compare it to \emph{bilateral trading}.
In standard bilateral trading, 
a single consumer wants to buy a physical product from a single producer; in our setting, \ua{} is the consumer, \ia{} is the producer, and the ``product'' is information. 
This is like bilateral trading, with the additional difficulty that the 
consumer cannot verify the ``product'' received.


The case when the production cost $c$ is unknown in bilateral trading was studied by 
\citet{Baron1982Regulating}. They 
define the \emph{virtual cost function} of the producer by:
\begin{align*}
z(c) := c + {H(c)\over h(c)} && (c\in[\cmin,\cmax])
\end{align*}
(it is analogous to the virtual valuation function used in Myerson's optimal auction theory).
By Myerson's theorem, the expected loss of the consumer in any truthful mechanism equals the expected virtual cost of the producer in that mechanism, $E_c[z(c)]$. The ``loss'' of the consumer from not buying the product is her utility from having this product, which we denote by $u$. 
Therefore, to minimize her expected loss, the consumer should buy the product if-and-only-if $z(c) < u$. 

Under standard regularity assumptions on $h$, the virtual cost function $z(c)$ is increasing with $c$. 
In that case, the optimal mechanism for the consumer is to make a take-it-or-leave-it offer to buy the product for a price of:
\begin{align}
\label{eq:offer}
R_{z,u} =
\begin{cases}
0 & \text{if~} u < z(\cmin)
\\
z^{-1}(u) & \text{if~} z(\cmin)\leq u\leq z(\cmax)
\\
\cmax & \text{if~} z(\cmax) < u 
\end{cases}
\end{align}
With this mechanism, the producer agrees to sell the product iff $c<R_{z,u}$, which occurs iff $z(c)<u$.

We now return to our original setting, in which the ``product'' is the information about an object's value.
We emphasize that there are two values: the value for both agents of the object itself, which we denoted by $v$, and the value for 
\ua{} of \emph{knowing} $v$, which we denote here by $u$.
We assume that the object's value $v$ and the cost $c$ of computing $v$ are independent random variables.

Similarly to the setting of \citet{Baron1982Regulating}, 
\ua{} has to ensure that \ia{} sells the information iff $z(c)<u$, which happens iff $c<R_{z,u}$.
Analogously to equations \eqref{eq:gc}, we define:
\begin{align*}
p_{z,u} &= {R_{z,u}\over  E_v[\max(0, v - E[v])]}
\\
G_{z,u}(r) &= p_{z,u}\cdot \step_{r>E[v]} 
+
(1-p_{z,u})\cdot \step_{r=\vmax}
\end{align*}
\Ua{} has to run Mechanism \ref{alg:2ndprice} using $G_{z,u}$ as the cdf. 
As explained after equations \eqref{eq:gc}, 
this gives \ia{} an expected net utility of $R_{z,u}$, so \ia{} will agree to participate in the mechanism iff $c < R_{z,u}$. This decision rule of the agent is in fact the one that maximizes the expected utility of \ua{}.
\begin{example}
Suppose the cost $c$ is distributed uniformly in $[0,\cmax]$. 
Then, the virtual cost function is $z(c) = 2 c$, so $z^{-1}(u)=u/2$ and $R_{z,u} = \min(\cmax,u/2)$.

Consider first a physical product.
If $u>2 \cmax$, then the consumer offers $\cmax$, the producer always sells, and the consumer's loss is $\cmax$. If $u<2 \cmax$, then the consumer offers $u/2$ and the producer sells only if the cost is less than $u/2$. This happens with probability ${u\over 2 \cmax}$, so the consumer's expected loss is ${u\over 2 \cmax}\cdot {u\over 2} + (1-{u\over 2 \cmax})\cdot u = u\cdot(1-{u\over 4\cmax})$, which is less than $\cmax$.

Now suppose that the ``product'' is information about an object's value. Suppose that, a~priori, the object's value is distributed uniformly in $[0,2 M]$.
As calculated in Example \ref{exm:applicability}, 
in this case $E[\max(0, v-E[v])] = M/4$. 
We make the realistic assumption that $\cmax < M/4$ (the maximum possible cost for appraising an object is less than a quarter of the expected value of the object). Therefore the following expression defines a valid probability:
\begin{align*}
p_{z,u} = {
\min(\cmax, u/2)
\over
M/4
}
\end{align*}
\Ua{} should run Mechanism \ref{alg:2ndprice} with the cdf
$G_{z,u}(r) = p_{z,u}\cdot \step_{r>E[v]} 
+
(1-p_{z,u})\cdot \step_{r=\vmax}$.
The expected net utility of \ia{} from participating is $\unet(G_{z,u}) = \min(\cmax, u/2)$.

If $u>2 \cmax$, then $\unet(G_{z,u}) = \cmax$, so \ia{} always participates, and \ua{} always obtains the information for an expected loss of $\cmax$.

If $u<2 \cmax$, then $\unet(G_{z,u}) = u/2$ and it might be higher or lower than the actual cost $c$. If $c<u/2$ then 
\ia{} participates and \ua{} obtains the information for an expected loss of $u/2$; if $c>u/2$ then \ia{} refuses to participate and \ua{} does not obtain the information, so her loss is $u$. 
All in all, \ua{}'s expected loss is ${u\over 2 \cmax}\cdot {u\over 2} + (1-{u\over 2 \cmax})\cdot u = u\cdot(1-{u\over 4\cmax})$, which is less than $\cmax$.
\end{example}

\section{Alternative Mechanisms}
\label{sec:alternatives}
In addition to Mechanism \ref{alg:2ndprice},
we developed two alternative mechanisms for solving the same problem --- eliciting a true value from a single information agent.

In Mechanism \ref{alg:1stprice} below, the price of the object is not determined by \ua{} but rather calculated as a function of \ia{}'s bid, similarly to a first-price auction.

\begin{algorithm}
\begin{enumerate}
\item \Ia{} bids $b \in [\vmin,\vmax]$.
\item With probability
$G(b)$%
, \ia{} 
buys the object for:
\begin{align*}
b - {\ghat(b)/ G(b)}
\end{align*}
where $\ghat$ is the integral of $G$: $\ghat(b) = \int_{r=\vmin}^b G(r) dr$.
\end{enumerate}
\caption{
\label{alg:1stprice}
Parameters: $\epsilon>0$: a constant, $G(r)$: a cdf.
}
\end{algorithm}

In Mechanism \ref{alg:postedprice} below, there is no bid at all --- \ua{} publicly posts a price and \ia{} decides whether to buy the object at this price or not. 

\begin{algorithm}
\begin{enumerate}
\item \Ua{} publicly posts the price $t$.
\item \Ua{} asks \ia{} whether he wants to buy the object for $t$ or not.
\item If \ia{} say yes, then with probability $p$ he buys the object from \ua{} for $t$. Otherwise the object is not sold.
\item If the object is not sold in step 3, then \ua{} runs Mechanism \ref{alg:2ndprice} with the function $G_0$ of equation \eqref{eq:g0}.
\end{enumerate}
\caption{
\label{alg:postedprice}
Parameters: $t>0$ --- a constant,
$p\in[0,1]$ --- a probability.
}
\end{algorithm}

The three mechanisms are apparently different in various aspects such as the role of the different players in the underlying Stackelberg game (leader vs. follower), and whether or not there is a requirement for secrecy (Mechanism \ref{alg:2ndprice} requires to keep $r$ secret while Mechanism \ref{alg:1stprice} need no secrecy).
Interestingly, they are equivalent in the conditions they impose on the cost $c$ and \ua{}'s loss.
We present a proof sketch below.  Note that we consider the general case of different valuations of \ia{}  and \ua{} (as in Subsection \ref{sub:different-values}) --- \ia{}'s value is $v_a$ and \ua{}'s value is $v_p$. 

In \textbf{Mechanism \ref{alg:1stprice}}, \ia{}'s expected utility for bidding $b$ is:
\begin{align*}
G(b)
\cdot(v_a - b + \ghat(b)/G(b))
\end{align*}
\Ia{} 
calculates the optimal bid $b$ by solving an optimization problem.
\ifdefined\ProveMechanismB
The derivative w.r.t. $b$ is:
\begin{align*}
&
G' \cdot (v_a - b + \ghat/G)+ G\cdot (-1 + (G^2-\ghat G')/(G^2))
\\
=
&
G'\cdot (v_a-b)
\end{align*}
Since $G$ is an increasing function, this expression is positive when $b<v_a$ and negative when $b>v_a$.
Therefore the expected utility of \ia{} is maximized by bidding $b = v_a$. In this case his expected utility is:
\begin{align*}
\ucoop(G) = 
E[G(v_a)\cdot (v_a - v_a + \ghat(v_a)/G(v_a))]
=
E[\ghat(v_a)]
\end{align*}
Similarly, when \ia{} does not compute $v_a$, his utility is optimized by bidding $b = E[v_a]$, which gives him a utility of:
\begin{align*}
\uheur(G) = \ghat(E[v_a])
\end{align*}
These utilities are exactly as in Mechanism \ref{alg:2ndprice} and Equation \eqref{eq:ucoop-uheur-unet}.
\else
A higher bid means a higher probability $G(b)$ of receiving the object but also a higher purchase price. If \ia{} knows $v_a$, then the solution to this optimization problem is $b = v_a$; otherwise it is $b = E[v_a]$. 
\Ia{}'s utilities when computing / not computing $v_a$ (i.e, the expressions $\ucoop, \uheur, \unet$) are exactly as in Mechanism \ref{alg:2ndprice} and Equation \eqref{eq:ucoop-uheur-unet}.
\fi
Therefore, Theorem \ref{trm:1} is valid as-is for Mechanism \ref{alg:1stprice}, and the mechanism is applicable iff $c < E[\max(0, v_a-E[v_a])]$, using the same cdf $G^*$ of \eqref{eq:g*}.

Moreover, \ua{}'s loss $\loss$ is exactly as in equation \eqref{eq:loss}. Therefore \ua{} has to solve the same optimization problem for minimizing the loss, and the minimal loss is the same. In particular, in the common value setting $v_a=v_p$, \ua{}'s loss can be made arbitrarily close to the information production cost $c$.

In \textbf{Mechanism \ref{alg:postedprice}}, 
in step 2, \ia{} has to decide whether to buy the object or not. Calculating the true value $v_a$ may help \ia{} decide:
\begin{itemize}
\item If \ia{} calculates $v_a$, he buys the object with probability $p$ iff $v_a > t$, so his utility is $E[p\cdot \max(0, v_a-t)]$.
\item Otherwise, he buys the object iff $E[v_a]>t$, so 
his expected utility is $p\cdot \max(0,E[v_a]-t)$.
\end{itemize}
In case \ia{} decides to calculate $v_a$, in step 3 the situation is similar to the $c=0$ case mentioned at the end of Section \ref{sub:min-loss} --- \ia{} already knows the information so the cost for calculating it \emph{now} is 0. Therefore, at step 4 \ua{} elicits the true information for almost zero additional loss.
Define the following functions (depending on the mechanism parameters $p,t$):
\begin{align*}
G(r) &:= p\cdot \step_{r>t}
\\
\ghat(v) &:= \int_{r=0}^v G(r) dr = p\cdot \max(0, v-t)
\end{align*}
With these definitions, 
\ia{}'s utilities when computing / not computing $v_a$ are exactly as in equation \eqref{eq:ucoop-uheur-unet}, and 
\ua{}'s loss in step 2 is the same as in equation \eqref{eq:loss}. So Theorem \ref{trm:1} is valid, and the mechanism is truthful iff $c < E[\max(0, v_a - E[v_a])]$, by taking $t=E[v_a]$ and $p>{c\over E[\max(0, v_a - E[v_a])]}$.

Additionally, \textbf{Mechanism \ref{alg:2ndprice}} itself can be extended by adding a probability of sale --- before actually selling the object to the agent for $r$, the principal tosses a biased coin with probability $p$ of success (where $p$ is a fixed parameter), and makes the sale only in case of success.
This extension is actually already supported by the current mechanism: for any cdf $G$ and parameter $p$, we can create a new cdf $G_p$ by putting a probability mass of $1-p$ on values larger than $\vmax$ and a probability mass of $p$ on the original $G$. This attains exactly the same effect as a sale with probability $p$, since with probability $1-p$, the $r$ will be so high that the object will never be sold. 
Hence, Theorem \ref{trm:1} applies to this extension too, so even with this generalization, the mechanism works iff $c < E[max(0,v-E[v])]$.

The fact that several different mechanisms lead to the same applicability conditions and loss expressions lead us to conjecture that these results are valid universally.

\begin{conjecture*}
	\label{conjecture}
	(a) There exists a mechanism for truthfully eliciting a single agent's value $v_a$, if-and-only-if:
	\begin{align*}
	c < E[\max(0, v_a-E[v_a])].
	\end{align*}
	
	(b) In any mechanism for truthfully eliciting $v_a$ from a single agent, 
	\ua{}'s loss is at least:
	\begin{align*}
	c + \min_{G\in \Gset_c}
	E[(v_p-v_a)\cdot G(v_a)],
	\end{align*}
	where the minimization is over all cumulative distribution functions satisfying $E[\ghat(v_a)] = c$.
\end{conjecture*}

While our main goal in showing three inherently different mechanisms is primarily theoretic (supporting our conjecture), there are some practical advantages for preferring the use of some of them in specific cases. For example, an advantage of Mechanism \ref{alg:1stprice} is that it does not require a mediator. Mechanism \ref{alg:2ndprice} requires a mediator to keep the reservation value $r$ secret and reveal it only after the agent's bid (and this recurs in Mechanism \ref{alg:postedprice} as it requires \ua{} to run also Mechanism \ref{alg:2ndprice} entirely). The mediator might collude with the agent and reveal the reservation value to him before the bid, allowing him to bid $r+\epsilon$ and win the object without giving any information. The mediator might also collude with the principal and reveal a false reservation value $b-\epsilon$ after hearing the bid $b$. In Mechanism \ref{alg:1stprice} no such problems arise.  

Additionally, Mechanism \ref{alg:2ndprice} requires \ia{} to believe that \ua{} really draws $r$ from the advertised distribution (or alternatively, use a third-party for doing the lottery). In Mechanism \ref{alg:1stprice} the lottery is much simpler: a probability $p=G(b)$ is calculated in a transparent manner, and the object is sold with probability $p$. Such lottery can be carried out transparently in front of \ia{}, so no trust is required. 

On the other hand,  Mechanism \ref{alg:2ndprice} has the advantage that the optimal strategy of bidding $v$ is more intuitive. In Mechanism \ref{alg:1stprice}, once the agent knows $v$, he needs to solve a complex optimization problem in order to calculate the optimal $b$ and realize that it equals $v$. In contrast, in Mechanism \ref{alg:2ndprice}, once the agent knows $v$, it is easy to realize that it is optimal to bid $v$: 
by bidding higher he might buy the object at a price higher than its value, and by bidding lower he might miss buying the object at a price lower than its value.

\section{Related Work}\label{sec:Related}
Mechanisms by which an uninformed agent tries to elicit information from an informed agent are as old as King Solomon's judgment (I Kings:3).
About two centuries ago, the German poet Goethe
invented a mechanism for eliciting the value of his new book from his publisher \citep{Moldovanu1998Goethes},
but without considering the computation cost.

Various new mechanisms for truthful information elicitation have been studied in recent years,  including mechanisms based on \emph{proper scoring rules} \citep{armantier2013,hossain2013binarized},
the 
\emph{Bayesian truth serum} \citep{prelec2004bayesian,barrage2010penny,weaver2013creating},
and its variants
\citep{offerman2009truth,witkowski2012robust}, 
the \emph{peer truth serum}
\citep{radanovic2016peer},
the \emph{ESP game}, credit-based mechanisms \cite{hajaj2015strategy}, and similar \emph{output agreement} mechanisms \citep{waggoner2013information}.
See also \citet{kong2019information} for a recent unifying framework for several different kinds of mechanisms.
\citet{faltings2017game} provide a comprehensive survey of such mechanisms in the computer science literature.
They classify them into two  categories:
\begin{quote}
The principle underlying all truthful mechanisms is to reward reports according to consistency with a reference. (1) In the case of \emph{verifiable} information, this reference is taken from the
ground truth as it will eventually be available. (2) In the case of \emph{unverifiable} information, it will be
constructed from peer reports provided by other agents.
\end{quote}
The present paper provides a third category: the information is unverifiable, and yet there is a single agent to elicit it from.
\ifdefined\FULLVERSION
\footnote{
When there are many agents, our problem becomes easier. For example, with three agents the following mechanism is possible: (a) Offer each agent to sell you the information for $c+\delta$. (b) Collect the reports of all agreeing agents. (c) If one report is not identical to at least one other report, then file a complaint against this agent and send her to jail.   This creates a coordination game where the focal point is to reveal the true value, like in the ESP game.  In our setting there is a single agent, so this trick is not possible.
}
\fi

Interactions between an informed agent and an uninformed principal have also been studied extensively in economics.
A typical setting is that \ia{} is a seller holding an object and \ua{} is a buyer wanting that object (contrary to our setting, where \ua{} is the object owner).
In some settings, \ia{} is a manager of a firm and \ua{} is a potential investor. Common to all cases is that \ia{} holds information that may affect the utility of \ua{}, and the question is if and how \ia{} can be induced to disclose this information.

The seminal work of \citet{akerlof1970market}
shows that, when information is not verifiable and is not guaranteed to be correct (as in our setting), the incentive of \ia{}  to provide false information might lead to complete market failure. 
In contrast, 
if the information is ex-post verifiable (i.e, \ia{} can hide information but cannot present false information), then market forces may be sufficient to push \ia{} to voluntarily disclose his information \citep{grossman1980disclosure,grossman1981informational,hajaj2017selective}.
Mechanisms for information elicitation have been developed for settings where the information is 
verifiable \citep{hart2016evidence},
partially verifiable \citep{green1986partially,glazer2004optimal} or verifiable at a cost \citep{ben2014optimal,moscarini2003optimal,wiegmann2010multi,emek}.

An additional line of research assumes that the information is unverifiable, however, if it is purchased, it is always correct. Moreover, their goal is to maximize \ia{}'s revenue rather than minimize \ua{}'s loss \cite{babaioff2012optimal,alkoby2017b,alkoby2015s,sarne2014c,alkoby2015strategicAmec}.

Our work is also related to \emph{contract theory} \citep{bolton2005contract}, in which a principal tries to incentivize an agent to choose  an action that is favorable to \ua{}.
There, while \ua{} cannot observe \ia{}'s action, she \emph{can} observe the (probabilistic) outcome of his action. In contrast, in our setting \ua{} has no way of knowing whether or not \ia{} calculated the true information.

Our work is motivated by government auctions for oil and gas fields. A lease for mining oil/gas is put to a first-price sealed-bid auction. 
One of the participating firms owns a nearby plot and can, by drilling in its own plot, compute relevant information about the potential value of the auctioned plot. 
\citet{hendricks1994auctions}
show that, in equilibrium, the informed firm underbids and gains information rent, while the uninformed firms have zero expected value.
\citet{porter1995role}
provides empirical evidence supporting this conclusion from almost 40 years of auctions by the US government. It indicates that information asymmetry causes the government to lose about 30\% of the potential revenue.
As a solution, \citet{hendricks1993optimal}
suggest to exclude \ia{} from the auction and induce him to reveal the information by promising him a fixed percentage of the auction revenues.
However, they note that in practice it may be impossible to exclude a firm from a government auction.
Our mechanism provides a different solution: the government (\ua{}) can use our mechanism to elicit the information from the informed firm (\ia{}). Then it can release the information to the other firms and by this remove the information asymmetry.

The decision rule in our Theorem \ref{trm:1} is somewhat similar to the ones used in optimal stopping problems, e.g., the one derived by \citet{weitzman1979optimal} for Pandora's Box problem.  While the latter considers a single player and has no strategic aspect, our model considers a strategic setting. Still the essence of the decision is somehow similar as it consider the marginal expected profit from knowing the true value (as opposed to acting based on the best value found so far in Pandora's problem, or the expected value in our case).

\ifdefined\FULLVERSION
Mechanism \ref{alg:2ndprice} uses the reservation-price concept, which can be found in literature on auctions where agents have to incur a cost for learning their own value. 
For example, 
\citet{hausch1993common}
 discuss an auction for a single item with a common value. Each bidder incurs a cost for participating in the auction, and additional cost for estimating the value of the item.
\citet{persico2000information}
discuss an extended model where the bidders can pay to make their estimate of the value more informative.
\else
The mechanisms we propose make use of the reservation-price concept, which can be found in literature on auctions where agents have to incur a cost for learning their own value \citep{hausch1993common,persico2000information}. 
\fi

\section{Discussion and Future Work}
\label{sec:future}
Information-providers are now ubiquitous, enabling people and agents to acquire 
information of various sorts. 
As self-interested agents, information providers typically seek to maximize their revenue. This is where the failure of direct payment becomes apparent, especially when the information provided is non-verifiable.  
The importance of the mechanisms provided and analyzed in the paper is therefore in their guarantee for truthfulness in the information elicitation process.

An important challenge for future work is to study the theoretic limitations of the setting studied in this paper --- a single information-agent and unverifiable information. 
In particular, we conjecture that any truthful mechanism for this setting must sell the object with a positive probability, although we did not yet prove this formally. 
Settling the conjecture in Section \ref{conjecture} is an interesting challenge too.
Some other directions for future work are:

\paragraph{Unknown distribution of value}
Mechanism \ref{alg:2ndprice} requires to calculate $E_v[\max(0, v-E[v])]$, which requires knowledge of the prior distribution of $v$.
When the distribution is not known, 
truthfulness can be guaranteed only 
when $c=0$, since in this case $G$ can be chosen independently of the distribution (see end of subsection \ref{sub:min-loss}).  It is interesting whether true information can be elicited by a prior-free mechanism also for $c>0$.

\paragraph{Risk-averse agents}
The current model assumes that \ia{} is risk-neutral, so that his utility from a random mechanism is the expectation of the value.
It is interesting to check what happens in the common case in which \ia{} is risk-averse.
\ifdefined\FULLVERSION
Suppose there is an increasing function $u:\mathbb{R}\to \mathbb{R}$ that maps the value of \ia{} to his utility. So $u(0)=0$ (no value means no utility) and $u'(x)>0$ for all $x$ (more value always means more utility). When the agent is risk-neutral (as we assumed so far), $u'$ is constant; when the agent is risk-averse, $u'$ is decreasing. 
Then, when \ia{} is cooperative and computes $v$, he gets the object iff $r \leq b(v)$, and then his utility is $u(v-r)$. Therefore his expected utility is:
\begin{align*}
\int_{r=0}^{b(v)} \bigg[(1-\epsilon)G'(r) + {\epsilon\over \vmax-\vmin}\bigg] u(v-r)dr
\end{align*}
Since $u(v-r)$ is positive iff $v-r$ is positive, the integrand is positive iff $r<v$. Therefore the expression is maximized when $b(v)=v$, and hence it is still a strictly dominant strategy for \ia{} to bid $b=v$. 
This is encouraging, since it means that at least the second part of Mechanism \ref{alg:2ndprice} (revealing the value after it is computed) remains truthful regardless of \ia{}'s risk attitude. 
\else
It is easy to show that the second part of Mechanism \ref{alg:2ndprice} (revealing the value after it is computed) remains truthful regardless of \ia{}'s risk attitude. 
\fi
However, computing \ia{}'s utility in the cooperative vs. the heuristic strategy is much harder.

\paragraph{Different effort levels}
In our setting, \ia{} has only two options: either calculate the accurate value, or not calculate it at all. When $c>E_{v}[max(0, v-E[v])]$, it is optimal for the agent to not calculate the value at all. Then,
it is optimal for the agent to bid $E[v]$. It is never optimal for the agent to bid an inaccurate value.

In more realistic settings, the agent may have three or more options. For example, it is possible that the agent can, by incurring a cost $c' < c$, get an inaccurate estimate of $v$ (e.g., the agent learns some value $u$ such that the true value is distributed uniformly in $[u-d,u+d]$, where $d$ is the inaccuracy parameter).  Then, the analysis of \ia{}'s behavior becomes more complex since there are more paths in which \ia{} may decide to calculate the true value: he may decide to incur the cost $c$ already from the start, or incur only the cost $c'$, and after observing the results --- decide whether to incur an additional cost of $c$. 
\Ua{}'s goal is to learn the true value at the end --- regardless of how many intermediate calculations are done by \ia{}. It is interesting to characterize the mechanisms that let \ua{} attain this goal.

\section{Acknowledgments}
This paper benefited a lot from discussions with the participants of the industrial engineering seminar in Ariel University, the game theory seminar in Bar Ilan University, the game theory seminar in the Hebrew University of Jerusalem and the Israeli artificial intelligence day.
We are particularly grateful to Sergiu Hart and Igal Milchtaich for their helpful mathematical ideas.
This research was partially supported by the Israel Science Foundation (grant No. 1162/17).

\newpage
\bibliographystyle{ACM-Reference-Format}
\balance  
\bibliography{information-elicitation}

\end{document}